\documentclass[prodmode]{acmsmall-ec15} 

\doi{2764468.2764522}

\usepackage{wrapfig}
\usepackage[ruled]{algorithm2e}

\SetAlFnt{\small}
\SetAlCapFnt{\small}
\SetAlCapNameFnt{\small}
\SetAlCapHSkip{0pt}
\IncMargin{-\parindent}

\usepackage[numbers]{natbib}

\usepackage[suppress]{color-edits}  
\addauthor{etn}{blue}
\addauthor{vsn}{red}
\addauthor{et}{black}     
\addauthor{vs}{black}    
\addauthor{dn}{black}   

\usepackage{definitions}

\newcommand{\OurTitle}{Econometrics for Learning Agents}

\begin{document}

\markboth{Nekipelov et al.}{\OurTitle}

\title{\OurTitle}
\author{DENIS NEKIPELOV
\affil{University of Virginia, {\tt denis@virginia.edu}}
VASILIS SYRGKANIS
\affil{Microsoft Research, {\tt vasy@microsoft.com}}
EVA TARDOS
\affil{Cornell University, {\tt eva.tardos@cornell.edu} }
}

\begin{abstract}
The main goal of this paper is to develop a theory of inference of player valuations from  observed data in the generalized second price auction without relying on the Nash equilibrium assumption. Existing work in Economics on inferring agent values from data relies on the assumption that all participant strategies are best responses of the observed play of other players, i.e. they constitute a Nash equilibrium. In this paper, we show how to perform inference relying on a weaker assumption instead: assuming that players are using some form of no-regret learning. Learning outcomes emerged in recent years as an attractive alternative to Nash equilibrium in analyzing game outcomes, modeling players who haven't reached a stable equilibrium, but rather use algorithmic learning, aiming to learn the best way to play from previous observations. In this paper we show how to infer values of players who use algorithmic learning strategies. Such inference is an important first step before we move to testing any learning theoretic behavioral model on auction data. We apply our techniques to a dataset from Microsoft's sponsored search ad auction system.
\end{abstract}



\keywords{Econometrics, no-regret learning, sponsored search, value inference}


\begin{bottomstuff}

Author's addresses: denis@virginia.edu, vasilis@cs.cornell.edu, eva.tardos@cornell.edu.
\'Eva Tardos was supported in part by NSF grants CCF-0910940; and CCF-1215994, ONR grant N00014-08-1-0031, a Yahoo! Research Alliance Grant, and a Google Research Grant.
\'Denis Nekipelov was supported in part by NSF grants CCF-1449239, CCF-0910940 and a Google Research Grant.
\end{bottomstuff}

\maketitle

\section{Introduction}\label{sec:intro}

\vsedit{The standard approach in the econometric analysis of strategic interactions, starts with assuming that the participating agents, conditional on their private parameters, such as their valuation in an auction, can fully optimize their utility given their opponents actions and that the system has arrived at a stable state of mutual such best-responses, aka a Nash equilibrium.}

In recent years, learning outcomes have emerged as an important alternative to Nash equilibria. This solution concept is especially compelling in online environments, such as Internet auctions, \vsedit{as many such environments are best thought of as repeated strategic interactions in a dynamic and complex environment, where participants need to constantly update their strategies to learn how to play optimally. 
} The strategies of agents in such environments evolve over time, as they learn to improve their strategies, and react to changes in the environment. \vsedit{Such learning behavior is reinforced even more by the increasing use of sophisticated algorithmic bidding tools by most high revenue/volume advertisers.}
With auctions having emerged as the main source of revenue on the Internet, there are multitudes of interesting data sets for strategic agent behavior in repeated auctions.

To be able to use the data on \vsedit{agent behavior} to empirically test the prediction of \etedit{the theory based on learning agents}, one first needs to infer \vsedit{the agents' types}, or valuations of the items, from the observed behavior \etedit{without relying on the stable state best-response assumption}. The idea behind inference \etedit{based on the stable best response assumption is straightforward: the  distribution of actions of players is observed in the data. If we assume that each player best responds to the distribution of opponents' actions, this best response can be recovered from the data. The best response function effectively describes the preferences of the players, and can typically be inverted to recover each player's private type. This is the idea used in
all existing work in economics on this inference problem,
including 
\cite{athey:nekipelov:2010}, \cite{Nekipelov:World:Congress}, and 
\cite{JiangLeyton-Brown}.
}

There are several caveats in this approach.
First of all, the assumption that an equilibrium has been reached
is unrealistic in many practical cases, either because equilibrium best
response functions are hard to compute or the amount of information
needed to be exchanged between the players to reach
an equilibrium outcome is unrealistically large.
Second, the equilibrium is rarely unique especially in dynamic settings.
In that case the simple ``inversion" of the best responses
to obtain the underlying valuations becomes a complicated
computational problem because the set of equilibria
frequently has a complicated topological structure.
The typical approach to avoid this complication
is to assume some equilibrium refinement, e.g.
the {\it Markov Perfect Equilibrium} in the case of
dynamic games.
\par
In spite of a common understanding in the Economics
community that many practical environments cannot be modeled
using the traditional equilibrium
framework (and that the assumption of a particular equilibrium refinement
holds), this point has been
overshadowed by  the simplicity of using the equilibrium based models.
\par
In this paper we consider a dynamic game \vsedit{where players learn how to play over time.} 
\vsedit{We focus on a model of {\it sponsored search auctions} where
bidders compete} for advertising spaces alongside the organic search results.
This environment is inherently dynamic where the auction is run for every individual
consumer search and thus each bidder with an active bid participates in a sequence
of auctions.
\par
Learning models agents that are new to the game, or participants as they adjust to a constantly evolving environment, where they have to constantly learn what may be best play. In these cases, strict best response to the observed past may not be a good strategy, it is both computationally expensive, and can be far from optimal in a dynamically changing environment.  Instead of strict best response, players may want to use a learning algorithm to choose their strategies.

No-regret learning \vsedit{
has emerged as an attractive} alternative to Nash equilibrium. The notion \vsedit{of a no-regret learning outcome,} generalizes Nash equilibrium by requiring the no-regret property of a Nash equilibrium, \vsedit{ in an approximate sense, but more importantly} without the assumption that player \vsedit{strategies are stable and independent.}  \vsedit{When considering 
a sequence of repeated plays, having no-regret means} that the total value for the player over a sequence of plays is not much worse than the value he/she would have obtained had he played the best single strategy throughout the time, where the best possible single strategy is determined with hindsight based on the environment and the actions of the other agents.
There are many well-known, natural no-regret learning algorithms, such as the weighted majority algorithm~\cite{Arora}, \cite{LittWarm94} (also known as Hedge~\cite{fre}) and regret matching \cite{HartMasColell} just to name a few simple ones. We propose \vsedit{
a theory of inference of agent valuations} just based on the assumption that the agent's learning strategies are smart enough that they have  minimal regret, without making any assumptions on the particular no-regret algorithms they employ.

There is a growing body of results in the algorithmic game theory literature \vsedit{
characterizing properties of no-regret learning outcomes in games, such as
approximate efficiency with respect to the welfare optimal outcome (see e.g. \cite{Roughgarden2012,Syrgkanis2013}).} 
For instance, 
\cite{C++JET13} consider the generalized second price auction in this framework, the auction also considered in our paper and showed that the average welfare of any no-regret learning outcome is always at least 30\% of the optimal welfare.
\vsedit{To be able to apply such theoretical results on real data and to quantify the true inefficiency of GSP in the real world under learning behavior, we first need a way to infer player valuations without relying on the equilibrium assumption.}

\paragraph{Our contribution}
The main goal of this paper is to develop \vsedit{a theory of value inference from observed data of repeated strategic interactions
without relying on a Nash equilibrium assumption.} Rather than relying on the stability of outcomes, we make the weaker assumption that players are using some form of no-regret learning. In a stable Nash equilibrium outcome, players choose strategies independently, and their choices are best response to the environment and the strategies of other players. Choosing a best response implies that the players have no regret for alternate strategy options. We make the analogous assumption, that in a sequence of play, players have \vsedit{small 
regret} for any fixed strategy.

Our results do not rely on the assumption that the participating players have correct beliefs
regarding the future actions of their opponents or that they even can correctly
compute expectations over the future values of state variables. Our only assumption is that the players understand the
rules of the game. This is significantly different from most of the current results in Industrial Organization
where estimation of dynamic games requires that the players have correct beliefs regarding
the actions of their opponents. This \etndelete{result} is especially important to the analysis of
bidder behavior in sponsored search auctions, which are \etnedit{the core} application 
of this paper. \etedit{The sponsored search marketplace is highly dynamic and volatile
where the popularity of different search terms is changing,
and the auction platform continuously runs experiments. In this environment advertisers continuously create new ads (corresponding to new bids), remove under-performing ads, learning what is the best way to bid while participating in the game.} In this setting the assumption
of \etndelete{forward-looking} players who have correct beliefs regarding their opponents and whose
bids constitute and equilibrium
may not be realistic.

When inferring player values from data, one needs to always accommodate small errors. In the context of players who employ learning strategies, a small error $\epsilon>0$ would mean that the player can have up to $\epsilon$ regret, i.e., the utility of the player from the strategies used needs to be at most $\epsilon$ worse than any fixed strategy with hindsight. Indeed, the guarantees provided by the standard learning algorithms is that the total utility for the player is not much worse than any single strategy with hindsight, with this error parameter decreasing over time, as the player spends more time learning. In aiming to infer the player's value from the observed data, we define the {\em rationalizable set} $\NR$, consisting of the set of values and error parameters $(v,\epsilon)$ such that with value $v$ the sequence of bid by the player would have at most $\epsilon$ regret. We show that $\NR$ is a closed convex set. Clearly, allowing for a larger error $\epsilon$ would result in a larger set of possible values $v$. The most reasonable error parameter $\epsilon$ for a player depends on the quality of his/her learning algorithm. We think of a rationalizable value $v$ for a player as a value that is rationalizable with a small enough $\epsilon$.

Our main result provides a characterization of the rationalizable set $\NR$ for the dynamic sponsored search auction game. We also provide a simple approach to compute this set. We demonstrate that the evaluation of $\NR$ is equivalent to the evaluation of a one-dimensional function which, in turn, can be computed from the auction data directly
by \etedit{sampling.} 
This result also allows us to show how much data is needed to correctly estimate the
rationalizable set for the dynamic sponsored search auction. We show that when $N$ auction samples
are observed in the data, the Hausdoff distance between the estimated and the true
sets $\NR$ is $O((N^{-1}\log\,N)^{\gamma/(2\gamma+1)})$, where
$\gamma$ is the sum of the number of derivatives of the allocation and pricing functions
and the degree of H\"{o}lder continuity of the oldest derivative.
In particular, when the allocation and pricing functions are only Lipschitz-continuous,
then the total number of derivatives is zero and the constant of H\"{o}lder-continuity
is $1$, leading to the rate $O((N^{-1}\log\,N)^{1/3})$.

This favorably compares our result to the result in \cite{athey:nekipelov:2010}
where under the assumption of {\it static Nash equilibrium} in the sponsored
search auction, the valuations of bidders can be evaluated with error
margin of order $O(N^{-1/3})$ when the empirical pricing and allocation
function are Lipschitz-continuous in the bid. This means that our approach,
which does not rely on the equilibrium properties, provides \etedit{convergence rate} 
which is only a $(\log\,N)^{1/3}$ factor away.

In Section \ref{sec:data} we test our methods on a Microsoft \etedit{Sponsored Search Auction} data set.
We show that our methods can be used to infer values on this real data, and
study the empirical consequences of our value estimation. We find that typical advertisers bid a significantly
shaded version of their value, shading it commonly by as much as 40\%. We also observe that each advertiser's account consists of a constant fraction of listings (e.g. bided keywords and ad pairs) that have tiny error and hence seem to satisfy the best-response assumption, whilst the remaining large fraction has an error which is spread out far from zero, thereby implying more that bidding on these listings is still in a learning transient phase.
Further, we find that, \vsedit{among listings that appear to be in the learning phase,
the relative error (the error in regret relative to the player's value) is slightly} positively correlated with the amount of shading. A higher error is suggestive of an exploration phase of learning, and is consistent with attempting larger shading of the bidder's value, while advertisers with smaller relative regret appear to shade their value less.

\paragraph{Further Related Work}
There is a rich literature in Economics that is devoted
to inference in auctions based on the equilibrium assumptions.
\cite{guerre2000} studies the estimation of values
in static first-price auctions, with the extension to the 
cases of risk-averse bidders in
\cite{guerre2009} and \cite{campo2011}. The
equilibrium settings also allow inference in 
the dynamic settings where the players participate
in a sequence of auctions, such as in \cite{jofre2003}.
A survey of approaches to inference is surveyed 
in \cite{athey2007}. These approaches have been
applied to the GSP and, more generally, to the 
sponsored search auctions in the empirical settings in 
\cite{Varian07} and \cite{athey:nekipelov:2010}.

The deviation from the ``equilibrium best response"
paradigm is much less common in any empirical studies.
A notable exception is \cite{haile2003} where the 
authors use two assumptions to bound the distribution
of values in the first-price auction. The first assumption 
is that the bid should not exceed the value. That allows
to bound the order statistics of the value distribution from 
below by the corresponding order statistics of the observed 
bid distribution. The second assumption is that there is
a minimum bid increment and if a bidder was outbid
by another bidder then her value is below the bid that made her
drop out by at least the bid increment.  That allows to provide
the bound from above. The issue with these bounds
is that in many practical instances they are very large and they 
do not directly translate to the dynamic settings. The computation
of the set of values compatible with the model may be compicated
even in the equilibrium settings. For instance, in \cite{aradillas2013}
the authors consider estimation of the distribution of values
in the ascending auction when the distribution of values may be
correlated. It turns out that even if we assume that the bidders
have perfect beliefs in the static model with correlated values,
the inference problem becomes computationally challenging.

\section{No-Regret Learning and Set Inference}\label{sec:general}

Consider a game $G$ \etnedit{with} 
a set $N$ of $n$ players. Each player $i$ has \etnedit{a} 
strategy space $B_i$. The utility of a player depends on the strategy profile $\vec{b}\in B_1\times \ldots\times B_n$, on a set of parameters $\theta\in \Theta$ that are observable in the data and on private parameters $v_i\in \V_i$ 
\etnedit{observable} only 
by each player $i$. We denote with $U_i(\vec{b}; \theta, v_i)$ the utility of a player $i$.

We consider a setting where game $G$ is played repeatedly. At each iteration $t$, each player picks a strategy $b_i^t$ and nature picks a set of parameters $\theta^t$. The private parameter $v_i$ of each player $i$ remains fixed throughout the sequence. We will denote with $\{\vec{b}\}_t$ the sequence of strategy profiles and with $\{\theta\}_t$ the sequence of nature's parameters. We assume that the sequence of strategy profiles $\{\vec{b}\}_t$ and the sequence of nature's parameters $\{\theta\}_t$ are observable in the data. However, the private parameter $v_i$ for each player $i$ is not observable.  The inference problem we consider in this paper is the problem of inferring these private values from the observable data.

We will refer to the two observable sequences as the \emph{sequence of play}. In order to be able to infer anything about agent's private values, we need to make some rationality assumption about the way agents choose their strategies. Classical work of inference assumes that each agents best response to the statistical properties of the observable environment and the strategies of other players, in other words assumes that game play is at a stable Nash equilibrium. In this paper, we replace this equilibrium assumption with the weaker {\em no-regret} assumption, stating that the utility obtained throughout the sequence is at least as high as any single strategy $b_i$ would have yielded, if played at every time step. If the play is stable throughout the sequence, no-regret is exactly the property required for the play to be at Nash equilibrium. However, no-regret can be reached by many natural learning algorithms without prior information, which makes this a natural rationally assumption for agents who are learning how to best play in a game while participating. More formally, we make no assumption of what agents do to learn, but rather will assume that agents learn well enough to satisfy the following no-regret property with a small error.

\etnedit{A sequence of play that we observe has $\epsilon_i$-regret for advertiser $i$ if:}
\begin{equation}\label{eqn:eps-regret}
\forall b'\in B_i: \frac{1}{T} \sum_{t=1}^{T} U_i\left(\vec{b}^t;\theta^t, v_i\right) \geq \frac{1}{T} \sum_{t=1}^{T} U_i\left(b',\vec{b}_{-i}^t;\theta^t, v_i\right)-\epsilon_i
\end{equation}

This leads to the following definition of a \emph{rationalizable set under no-regret learning}.
\begin{definition}[Rationalizable Set]
A pair $(\epsilon_i,v_i)$ of a value $v_i$ and error $\epsilon_i$ is a rationalizable pair for player $i$ if it satisfies Equation \eqref{eqn:eps-regret}. We refer to the set of such pairs as the \emph{rationalizable set} and denote it with $\NR$.
\end{definition}

The rationality assumption of the inequality (\ref{eqn:eps-regret}) models players who may be learning from the experience while participating in the game. We assume that the strategies $b_i^t$ and nature's parameters $\theta^t$ are picked simultaneously, so agent $i$ cannot pick his strategy dependent on the state of nature $\theta^t$ or the strategies of other agents $b_{i-1}^t$. This makes the standard of a single best strategy $b_i$ natural, as chosen strategies cannot depend on $\theta^t$ or $b_{i-1}^t$. Beyond this, we do not make any assumption on what information is available for the agents, and how they choose their strategies. Some learning algorithms achieve this standard of learning with very minimal feedback (only the value experienced by the agent as he/she participates). If the agents know a distribution of possible nature parameters $\theta$ \etnedit{ or is}
able to observe the past values of the parameters $\theta^t$ or the strategies $b_{i-1}^t$ (or both), and then they can use this observed past information to select their strategy at time step $t$. Such additional information is clearly useful in speeding up the learning process for agents. We will make no assumption on what information is available for agents for learning, or what algorithms they may be using to update their strategies. We will simply assume that they use algorithms that achieve the no-regret (small regret) standard expressed in inequality (\ref{eqn:eps-regret}).

For general games and general private parameter spaces, the \emph{rationalizable set} can be an arbitrary set with no good statistical or convexity properties. Our main result is to show that for the game of sponsored search auction we are studying in this paper, the set is convex and has good convergence properties in terms of estimation error. 

\section{Sponsored Search Auctions Model}\label{sec:auction-model}

We consider data generated by advertisers repeatedly participating in sponsored search auction. The game $G$ that is being repeated at each stage is an instance of a generalized second price auction triggered by a search query.

The rules of each auction are as follows\footnote{ignoring small details that we will ignore and are rarely in effect}: Each advertiser $i$ is associated with a click probability $\gamma_i$ and a scoring coefficient $s_i$ and is asked to submit a bid-per-click $b_i$. Advertisers are ranked by their rank-score $q_i=s_i\cdot b_i$ and allocated positions in decreasing order of rank-score as long as they pass a rank-score reserve $r$.  If advertisers also pass a higher mainline reserve $m$, then they may be allocated in the positions that appear in the mainline part of the page, but at most $k$ advertisers are placed on the mainline.

If advertiser $i$ is allocated position $j$, then he is clicked with some probability $p_{ij}$, which we will assume to be separable into a part $\alpha_j$ depending on the position and a part $\gamma_i$ depending on the advertiser, and that the position related effect is the same in all the participating auctions:
\begin{equation}
p_{ij} = \alpha_j\cdot \gamma_i
\end{equation}
We denote with $\vec{\alpha}=(\alpha_1,\ldots,\alpha_m)$ the vector of position coefficients. All the mentioned sets of parameters $\theta=(\vec{s}, \vec{b}, \gamma, r,m,k,\vec{\alpha})$ are observable in the data.

If advertiser $i$  is allocated position $j$, then he pays only when he is clicked and his payment, i.e. his cost-per-click (CPC) is the minimal bid he had to place to keep his position, which is:
\begin{equation}
c_{ij}(\vec{b}; \theta) = \max\left\{\frac{s_{\pi(j+1)}\cdot b_{\pi(j+1)}}{s_{i}}, \frac{r}{s_i},\frac{m}{s_i}\cdot {\bf 1}\{j\in M\}\right\}
\end{equation}
where by $\pi(j)$ we denote the advertiser that is allocated position $j$ and with $M$ we denote the set of mainline positions.

We also assume that each advertiser has a value-per-click (VPC) $v_i$, which is not observed in the data. If under a bid profile $\vec{b}$, advertiser $i$ is allocated slot $\sigma_i(\vec{b})$, his expected utility is:
\begin{equation}
U_i(b;\theta, v_i) = \alpha_{\sigma_i(b)}\cdot \gamma_i \cdot \left( v_i - c_{i\sigma_i(\vec{b})}(\vec{b};\theta)\right)
\end{equation}
We will denote with:
\begin{equation}
P_i(\vec{b};\theta) = \alpha_{\sigma_i(\vec{b})}\cdot \gamma_i
\end{equation}
the probability of a click as a function of the bid profile and with:
\begin{equation}
C_i(\vec{b}; \theta) = \alpha_{\sigma_i(\vec{b})}\cdot \gamma_i \cdot c_{i\sigma_i(\vec{b})}(\vec{b};\theta)
\end{equation}
the expected payment as a function of the bid profile. Then the utility of advertiser $i$ at each auction is:
\begin{equation}
U_i(\vec{b};\theta, v_i) = v_i \cdot P_i(\vec{b})-C_i(\vec{b})
\end{equation}

The latter fits in the general repeated game framework, where the strategy space $B_i=\R_+$ of each player $i$ is simply any non-negative real number.  The private parameter of a player $v_i$ is an advertiser's value-per-click (VPC) and the set of parameters that affect a player's utility at each auction and are observable in the data is $\theta$. At each auction $t$ in the sequence  the observable parameters $\theta^t$ can take arbitrary values that depend on the specific auction. However, we assume that the per-click value of the advertiser remains fixed throughout the sequence.

\paragraph{Batch Auctions} Rather than viewing a single auction as a game that is being repeated, we will view a batch of many auctions as the game that is repeated in each stage. This choice is reasonable, as it is impossible for advertisers to update their bid after each auction. Thus the utility of a single stage of the game is simply the average of the utilities of all the auctions that the player participated in during this time period. Another way to view this setting is that the parameter $\theta$ at each iteration $t$ is not deterministic but rather is drawn from some distribution $D^t$ and a player's utility at each iteration is the expected utility over $D^t$. In effect, the distribution $D^t$ is the observable parameter, and utility depends on the distribution, and not only on a single draw from the distribution. With this in mind, the per-period probability of click is
\begin{equation}
\textstyle{P_i^t(\vec{b}^t)=\E_{\theta\sim D^t}[P_i(\vec{b}^t;\theta)]}
\end{equation}
and the per-period cost is
\begin{equation}
\textstyle{C_i^t(\vec{b}^t) = \E_{\theta\sim D^t}[C_i(\vec{b}^t;\theta)].}
\end{equation}

We will further assume that the volume of traffic is the same in each batch, so the average utility of an agent over a sequence of stages is expressed by
\begin{equation}
 \frac{1}{T} \sum_{t=1}^{T}\left( v_i\cdot P_i^t(\vec{b}^t) - C_i^t(\vec{b}^t)\right).
\end{equation}

Due to the large volume of auctions that can take place in between these time-periods of bid changes, it is sometimes impossible to consider all the auctions and calculate the true empirical distribution $D^t$. Instead it is more reasonable to approximate $D^t$, by taking a sub-sample of the auctions. This would lead only to statistical estimates of both of these quantities. We will denote these estimates by $\hat{P}_i^t(\vec{b})$ and $\hat{C}_i^t(\vec{b})$ respectively. We will analyze the statistical properties of the estimated rationalizable set under such subsampling in Section \ref{sec:statistical}. 

\section{Properties of Rationalizable Set for Sponsored Search Auctions}\label{sec:convexity}

For the auction model that we are interested in, Equation \eqref{eqn:eps-regret} that determines whether a pair $(\epsilon,v)$ is rationalizable boils down to:
\begin{equation}
\forall b'\in \R_+: v\cdot \frac{1}{T} \sum_{t=1}^{T}\left( P_i^t(b',\vec{b}_{-i}^t)-P_i^t(\vec{b}^t) \right) \leq \frac{1}{T} \sum_{t=1}^{T} \left(C_i^t(b',\vec{b}_{-i}^t)-C_i^t(\vec{b}^t)\right) +\epsilon
\end{equation}
If we denote with 
\begin{equation}
\Delta P(b')= \frac{1}{T} \sum_{t=1}^{T}\left(  P_i^t(b',\vec{b}_{-i}^t)-P_i^t(\vec{b}^t)\right),
\end{equation}
the increase in the average probability of click from player $i$  switching to a fixed alternate bid $b'$ and with
\begin{equation}
\Delta C(b')=\frac{1}{T} \sum_{t=1}^{T} \left(C_i^t(b',\vec{b}_{-i}^t)-C_i^t(\vec{b}^t)\right),
\end{equation}
the increase in the average payment from player $i$  switching to a fixed alternate bid $b'$, then the condition simply states:
\begin{equation}\label{eqn:halfplanes}
\forall b'\in \R_+: v\cdot \Delta P(b') \leq \Delta C(b') + \epsilon
\end{equation}
Hence, the rationalizable set $\NR$ is an
envelope of the family of half planes obtained by varying $b \in {\mathbb R}_+$ in Equation \eqref{eqn:halfplanes}.

We conclude this section by characterizing some basic properties of this set, which will be useful both in analyzing its statistical estimation properties and in computing the empirical estimate of $\NR$ from the data.
\begin{lemma}
The set $\NR$ is a closed convex set.
\end{lemma}
\begin{proof}
The Lemma follows by the fact that $\NR$ is defined by a set of linear constraints. Any convex combination of two points in the set will also satisfy these constraints, by taking the convex combination of the constraints that the two points have to satisfy. The closedness follows from the fact that points that satisfy the constraints with equality are included in the set.
\end{proof}

\begin{lemma}\label{lem:upper-lower}
For any error level $\epsilon$, the set of values that belong to the rationalizable set is characterized by the interval:
\begin{equation}
v\in \left[ \max_{b': \Delta P(b')< 0}  \frac{\Delta C(b')+\epsilon}{\Delta P(b')},  \min_{b': \Delta P(b') > 0} \frac{\Delta C(b')+\epsilon}{\Delta P(b')}\right]
\end{equation}
In particular the data are not rationalizable for error level $\epsilon$ if the latter interval is empty.
\end{lemma}
\begin{proof}
Follows from Equation \eqref{eqn:halfplanes}, by dividing with $\delta P(b')$ and taking cases on whether $\Delta P(b')$ is positive or negative.
\end{proof}

To be able to estimate the rationalizable set $\NR$ we will need to make a few additional assumptions.
All assumptions are natural, and likely satisfied by any data. Further, the assumptions are easy to test for in the data, and the parameters needed can be estimated.

Since $\NR$ is a closed convex set, it is conveniently represented by the set of support
hyperplanes defined by the functions $P_i^t(\cdot,\cdot)$ and $C^t_i(\cdot,\cdot)$. Our first assumption is on the functions $P_i^t(\cdot,\cdot)$ and $C^t_i(\cdot,\cdot)$.

\begin{assumption}\label{assume1}
The support of bids is a compact set $B=[0,\overline{b}]$.
For each bidvector $\vec{b_{-i}^t}$ 
the functions $P_i^t(\cdot,\vec{b}_{-i}^t)$ and $C^t_i(\cdot,\vec{b}_{-i}^t)$
are continous,  monotone increasing 
and bounded on $B$.
\end{assumption}

\paragraph{Remark} These assumptions are natural and are satisfied by our data. In most applications, it is easy to have a simple upper bound $M$ on the maximum conceivable value an advertiser can have for a click, and with this maximum, we can assume that the bids are also limited to this maximum, so can use $B=M$ as a bound for the maximum bid. The probability of getting a click as well as the cost per-click are clearly increasing function of the agent's bid in the sponsored search auctions considered in this paper, \etedit{and the continuity of these functions is a good model in large or uncertain environments.}

We note that properties in Assumption \ref{assume1} implies that the same is also satisfied by
the linear combinations of functions, as a linear combination of monotone functions is monotone, implying that
the functions $\Delta P(b)$ and $\Delta C(b)$ are also monotone \etedit{and continuous}.

The assumption further implies that for any value $v$, there exists at least one
element of $B$ that maximizes $v\Delta P(b)-\Delta C(b)$, \etedit{as $v\Delta P(b)-\Delta C(b)$ is a
continuous function on a compact set}.

\subsection{Properties of the Boundary}

Now we can study the properties of the set $\NR$ by characterizing its boundary, denoted
$\partial \NR$.
\begin{lemma}
Under Assumption \ref{assume1},
$\partial \NR=\{(v,\epsilon)\,:\,\sup_b\left(v\Delta P(b)-\Delta C(b)\right)=\epsilon\}
$
\end{lemma}
\begin{proof}
(1) Suppose that $(v,\epsilon)$ solves $\sup_b\left(v\Delta P(b)-\Delta C(b)\right)=\epsilon$.
Provided the continuousness of functions $\Delta P(b)$ and $\Delta C(b)$ the supremum exists
and it is attained at some $b^*$ in the support of bids.
Take some $\delta>0$. Taking point $(v',\epsilon')$ such that $v'=v$ and $\epsilon'=\epsilon-\delta$
ensures that $v'\Delta P(b^*)-\Delta C(b^*)>\epsilon'$ and thus $(v'\epsilon') \not\in \NR$.
Taking the point $(v^{\prime\prime},\epsilon^{\prime\prime})$ such that
$v^{\prime\prime}=v$ and $\epsilon^{\prime\prime}=\epsilon+\delta$ ensures that
$$
\sup_b\left(v^{\prime\prime}\Delta P(b)-\Delta C(b)\right)<\epsilon^{\prime\prime},
$$
and thus for any bid $b$: $v^{\prime\prime}\Delta P(b)-\Delta C(b)<\epsilon^{\prime\prime}$. Therefore
$(v^{\prime\prime},\epsilon^{\prime\prime}) \in \NR$.
Provided that $\delta$ was arbitrary, this ensures that in any any neighborhood of $(v,\epsilon)$
there are points that both belong to set $\NR$ and that do not belong to it. This means that this
point belongs to $\partial \NR$.\\
(2) We prove the sufficiency part by contradiction. Suppose that $(v,\epsilon) \in \partial \NR$ and
$\sup_b\left(v\Delta P(b)-\Delta C(b)\right)\neq \epsilon$.
If $\sup_b\left(v\Delta P(b)-\Delta C(b)\right)> \epsilon$, provided that the objective
function is bounded and has a compact support, there exists a point $b^*$ where the supremum is attained.
For this point $v\Delta P(b^*)-\Delta C(b^*)> \epsilon$. This means that for this
bid the imposed inequality constraint is not satisfied and this point
does not belong to $\NR$, which is the contradiction to our assumption.\\
Suppose then that $\sup_b\left(v\Delta P(b)-\Delta C(b)\right)< \epsilon$.
Let $\Delta \epsilon=\epsilon-\sup_b\left(v\Delta P(b)-\Delta C(b)\right)$.
Take some $\delta<\Delta \epsilon$ and take an arbitrary $\delta_1,\delta_2$ be such that $\max\left\{
|\delta_1|,\frac{|\delta_2|}{\sup_b\Delta P(b)}\right\}<\frac{\delta}{2}$.
Let $v'=v+\delta_1$ and $\epsilon'=\epsilon+\delta_2$. Provided that
$$
\sup_b(v\Delta P(b)-\Delta C(b)) \geq \sup_b((v+\delta_1)\Delta P(b)-\Delta C(b))-\delta_1 \sup_b\Delta P(b)
\delta_1,
$$
for any $\tau \in (-\infty,\Delta\epsilon)$
we have
$$
\sup_b(v'\Delta P(b)-\Delta C(b))< \epsilon-\tau+ \delta_1 \sup_b\Delta P(b).
$$
Provided that $|\delta_1 \sup_b\Delta P(b)|<\delta/2 < \Delta\epsilon/2$, for any $|\delta_2|<\delta/2$,
we can find $\tau$ with $-\tau+ \delta_1 \sup_b\Delta P(b)=\delta_2$.
Therefore
$$
\sup_b(v'\Delta P(b)-\Delta C(b))< \epsilon',
$$
and thus for any $b$: $v'\Delta P(b)-\Delta C(b)< \epsilon'$.
Therefore $(v',\epsilon') \in \NR$.
This means that for $(v,\epsilon)$ we constructed a neighborhood
of size $\delta$ such that each element of that neighborhood is
an element of $\NR$. Thus $(v,\epsilon)$ cannot belong to $\partial \NR$.
\end{proof}
The next step will be to establish the properties of the boundary
by characterizing the solutions of the optimization
problem  of selecting the optimal bid single $b$ with for a given value $v$ and the sequence of bids by other agents, defined by $\sup_b\left(v\Delta P(b)-\Delta C(b)\right)$.
\begin{lemma}\label{monotone:comparative:statics}
Let $b^*(v)=\mbox{arg}\sup_b\left(v\Delta P(b)-\Delta C(b)\right)$. Then $b^*(\cdot)$
is upper-hemicontinuous and monotone.
\end{lemma}
\begin{proof}
\etedit{By Assumption (\ref{assume1})} the function $v\Delta P(b)-\Delta C(b)$ is continuous
in $v$ and $b$, then the upper hemicontinuity of $b^*(\cdot)$
follows directly from the Berge's maximum theorem.
\par
\etedit{To show that $b^*(\cdot)$ is monotone} consider the function $q(b;v)=v\Delta P(b)-\Delta C(b)$.
We'll show that this
function is {\it supermodular} in $(b;v)$, that is, for $b'>b$ and $v'>v$ we have
$$
q(b';v')+q(b;v)\ge q(b';v)+q(b;v').
$$
To see this observe that if we take $v'>v$, then
$$
q(b;v')-q(b;v)=(v'-v)\Delta P(b),
$$
which is non-decreasing in $b$ due to the monotonicity of $\Delta P(\cdot)$, implying that $q$ is supermodular.
Now we can apply the \etedit{Topkis' (\cite{topkis}, \cite{vives})}
monotonicity theorem from which we immediately conclude
that $b^*(\cdot)$ is non-decreasing.
\end{proof}
Lemma \ref{monotone:comparative:statics} provides us
a powerful result of the monotonicity of the
optimal response  function $b^*(v)$ which only relies
on the boundedness, compact support and monotonicity of
functions $\Delta P(\cdot)$
and $\Delta C(\cdot)$ {\it without relying on their differentiability}.
The downside of this is that $b^*(\cdot)$ can be
set valued or discontinuously change in the value.
To avoid this we impose the following additional assumption.

\begin{assumption}\label{assume2}
For each $b_1$ and $b_2>0$ the {\it incremental cost per click}
function
$$
\mbox{ICC}(b_2,b_1)=\frac{\Delta C(b_2)-\Delta C(b_1)}{\Delta P(b_2)-\Delta P(b_1)}
$$
is continuous in $b_1$ \dnedit{for each $b_2 \neq b_1$ and it is continuous in $b_2$
for each $b_1 \neq b_2$.}\dndelete{and $b_2$ (provided that $b_2 \neq b_1$) and bounded away from zero.}
Moreover for any $b_4>b_3>b_2>b_1$ on $B$:
$
\mbox{ICC}(b_4,b_3) > \mbox{ICC}(b_2,b_1).
$
\end{assumption}
\paragraph{Remark} Assumption \ref{assume2} requires that there are no discounts
on buying clicks ``in bulk": for any average position of the bidder
in the sponsored search auction, an incremental increase in
the click yield requires the corresponding increase in the payment.
In other words, ``there is no free clicks". In \cite{athey:nekipelov:2010}
it was shown that in sponsored auctions with uncertainty and a
reserve price, the condition in Assumption \ref{assume2}
is satisfied with the lower bound on the incremental cost
per click being the maximum of the bid and the reserve price (conditional
on the bidder's participation in the auction).

\begin{lemma}\label{continuity}
Under Assumptions \ref{assume1} and \ref{assume2}
$b^*(v)=\mbox{arg}\sup_b\left(v\Delta P(b)-\Delta C(b)\right)$
is a continuous monotone function.
\end{lemma}
\begin{proof}
Under Assumption \ref{assume1} in Lemma \ref{monotone:comparative:statics}
we established the monotonicity and upper-hemicontinuity of the
mapping $b^*(\cdot)$. We now strengthen this result to
monotonicity and continuity.
\par
Suppose that for value $v$, function $v\Delta P(b)-\Delta C(b)$
has an interior maximum and let $B^*$ be its set of maximizers.
First, we show that under Assumption \ref{assume2} $B^*$
is singleton. In fact, suppose that $b_1^*,b_2^* \in B^*$
and wlog $b_1^*>b_2^*$. Suppose that $b \in (b_1^*,b_2^*)$.
First, note that $b$ cannot belong to $B^*$. If it does, then
$$
v\Delta P(b^*_1)-\Delta C(b^*_1)=v\Delta P(b)-\Delta C(b)=v\Delta P(b^*_2)-\Delta C(b^*_2),
$$
and thus $\mbox{ICC}(b,b_1^*)=\mbox{ICC}(b^*_2,b)$
for $b_1^*<b<b^*_2$ which violates Assumption \ref{assume2}.
Second if $b \not\in B^*$, then
$$
v\Delta P(b^*_1)-\Delta C(b^*_1)\geq v\Delta P(b)-\Delta C(b),
$$
and thus $v \leq \mbox{ICC}(b,b_1^*)$. At the same time,
$$
v\Delta P(b^*_2)-\Delta C(b^*_2)\geq v\Delta P(b)-\Delta C(b),
$$
and thus $v \geq \mbox{ICC}(b_2^*,b)$. This is impossible
under Assumption \ref{assume2} since it requires that
$\mbox{ICC}(b_2^*,b)>\mbox{ICC}(b,b_1^*)$. Therefore, this means
that $B^*$ is singleton.
\par
Now consider $v$ and $v'=v+\delta$ for a sufficiently small
$\delta>0$ and $v$ and $v'$ leading to the interior maximum of the
of the objective function. By the result of Lemma \ref{monotone:comparative:statics},
$b^*(v') \geq b^*(v)$ and by Assumption \ref{assume2} the inequality
is strict. \dnedit{Next note that for any $b>b^*(v)$ $\mbox{ICC}(b,b^*(v))>v$
and for any $b<b^*(v)$ $\mbox{ICC}(b,b^*(v))<v$. By continuity, $b^*(v)$
solves $\mbox{ICC}(b,b^*(v))=v$. Now let $b'$ solve $\mbox{ICC}(b',b^*(v))=v'$.
By Assumption \ref{assume2}, $b'>b^*(v')$ since $b^*(v')$ solves $\mbox{ICC}(b,b^*(v'))=v'$
and $b^*(v')>b^*(v)$. 
Then by Assumption \ref{assume2}, $\mbox{ICC}(b',b^*(v))-v>v-\mbox{ICC}(0,b^*(v))$.
This means that $b'-b^*(v)<(v'-v)/(v-\mbox{ICC}(0,b^*(v)))$. This means that 
$|b^*(v')-b^*(v)|<|b'-b^*(v)|<\delta/(v-\mbox{ICC}(0,b^*(v)))$. This means that 
$b^*(v)$ is continuous in $v$.}
\end{proof}

Lemma \ref{continuity}, the interior solutions
maximizing $v\Delta P(b)-\Delta C(b)$ are continuous
in $v$. We proved that each boundary point of $\NR$
corresponds to the maximant of this function. As a result,
whenever the maximum is interior, then the boundary
shares a point with the support hyperplane
$v\Delta P(b^*(v))-\Delta C(b^*(v))=\epsilon$. Therefore,
the corresponding normal vector at that boundary
point is $(P(b^*(v)),-1)$.

\subsection{Support Function Representation}

Our next step would be to use the derived properties of the
boundary of the set $\NR$ to compute it. The basic idea will be based
on varying $v$ and computing the bid  $b^*(v)$ that maximizes regret. The corresponding
maximum value will deliver the value of $\epsilon$ corresponding to the boundary point.
Provided that we are eventually interested in the
properties of the set $\NR$ (and the quality of its approximation
by the data), the characterization via the support hyperplanes will
not be convenient because this characterization requires to solve
the computational problem of finding an envelope function for the
set of support hyperplanes. Since closed convex bounded sets
are fully characterized by their boundaries, we can use
the notion of the {\it support} function to represent the
boundary of the set $\NR$. 
\begin{definition}
The support function of a closed convex set $X$ is defined as:
$$
h(X,u)=\sup_{x \in X}\langle x,u\rangle,
$$
where in our case $X=\NR$ is a subset of ${\mathbb R}^2$ or value and error pairs $(v,\epsilon)$, and  then $u$ is also an element of ${\mathbb R}^2$.
\end{definition}

An important property of the support function
is the way it characterizes closed convex bounded sets. Recall that the Hausdorf norm for subsets
$A$ and $B$ of the metric space $E$ with metric $\rho(\cdot,\cdot)$ is defined
as
$$
d_H(A,B)=\max\{\sup\limits_{a \in A}\inf\limits_{b \in B}\rho(a,b),\,
\sup\limits_{b \in B}\inf\limits_{a \in A}\rho(a,b)\}.
$$
It turns out that $d_H(A,B)=\sup_u|h(A,u)-h(B,u)|$.
Therefore, if we find $h(\NR,u)$, this will be equivalent
to characterizing $\NR$ itself.

We note that the set $\NR$ is not bounded: the larger
is the error $\epsilon$ that the player can make, the
wider is the range of values that will rationalize the
observed bids. We consider the restriction of the
set $\NR$ to the set where we bound the error
for the players. Moreover, we explicitly assume
that the values per click of bidders have to
be non-negative.

\begin{assumption}\label{assume3}
\begin{itemize}
\item[(i)] The rationalizable values per click are non-negative: $v \geq 0$.
\item[(ii)]There exists a global upper bound $\bar \epsilon$ for the
error of all players.
\end{itemize}
\end{assumption}
\paragraph{Remark}
While non-negativity of the value may be straightforward
in many auction context, the assumption of an upper
bound for the error may be less straightforward.
\etedit{One way to get such an error bound is if we assume that the values come from a bounded range $[0,M]$ (see Remark after Assumption (\ref{assume1})), and then an error $\epsilon >M$ would correspond to negative utility by the player, and the players may choose to exit from the market
if their advertising results in negative value. }

\begin{theorem}
Under Assumption \ref{assume1}, \ref{assume2} the support function
of $\NR$ is
$$
h(\NR,u)=\left\{
\begin{array}{ll}
+\infty,&\,\mbox{if}\,u_2 \geq 0,\,\mbox{or}\, \frac{u_1}{|u_2|}\not\in \left[\inf_b \Delta P(b),\,\sup_{b}\Delta P(b)\right],\\
|u_2| \Delta C\left(\Delta P^{-1}\left(\frac{u_1}{|u_2|}\right)\right),&\,\mbox{if}\,
u_2<0\,\mbox{and}\,\frac{u_1}{|u_2|}\in \left[\inf_b \Delta P(b),\,\sup_{b}\Delta P(b)\right].\\
\end{array}
\right.
$$
\end{theorem}
\begin{proof}
Provided that the support function is positive homogenenous,
without loss of generality we can set $u=(u_1,u_2)$ with $\|u\|=1$.
To find the support function, we take $u_1$ to be dual to $v_i$
and $u_2$ to be dual to $\epsilon_i$. We re-write the inequality
of the half-plane as: 
$
v_i \cdot \Delta P(b_i) -\epsilon_i\leq \Delta C(b_i)
$.
We need to evaluate the inner-product
$$
u_1v_i+u_2\epsilon_i
$$
from above. We note that to evaluate the support
function for $u_2 \geq 0$, the corresponding inequality
for the half-plane needs to ``flip" to $\geq -\Delta C(b_i)$.
This means that for $u_2 \geq 0$, $h(\NR,u)=+\infty$. Next,
for any $u_2 < 0$ we note that the inequality for the half-plane
can be re-written as:
$
v_i |u_2| \Delta P(b_i) +u_2\epsilon_i\leq |u_2|\Delta C(b_i)$.

As a result if there exists $b_i$ such that for a given $u_1$: 
$
\frac{u_1}{|u_2|}=\Delta P(b_i)$,
then $|u_2|\Delta C(b_i)$ corresponds to the support function for this $b_i$.

Now suppose that $\sup_b\Delta P(b)>0$ and $u_1>|u_2|\sup_b\Delta P(b)$.
In this case we can evaluate
\begin{align*}
u_1v_i+u_2\epsilon_i = (u_1-|u_2|\sup_b\Delta P(b))v_i+|u_2|\sup_b\Delta P(b)v_i
+u_2\epsilon_i.
\end{align*}
Function $|u_2|\sup_b\Delta P(b)v_i
+u_2\epsilon_i$ is bounded by $|u_2|\sup_b\Delta C(b)$ for each $(v_i,\epsilon_i) \in \NR$.
At the same time, function $(u_1-|u_2|\sup_b\Delta P(b))v_i$ is strictly increasing
in $v_i$ on $\NR$. As a result, the support function evaluated at any vector
$u$ with $u_1/|u_2|>\sup_b\Delta P(b)$ is $h(\NR,u)=+\infty$.
\end{proof}

This behavior of the support function can be explained intuitively.
The set $\NR$ is a convex set in $\epsilon>0$ half-space. The
unit-length $u$ corresponds to the normal vector to the boundary
of $\NR$ and $-h(\NR,u)$ is the point of intersection of the $\epsilon$
axis by the tangent line. Asymptotically, the boundaries of the set
will approach to the lines $v_i \sup_b\Delta P(b)-\epsilon \leq \sup_b \Delta C(b)$
and $v_i \inf_b\Delta P(b)-\epsilon \leq \inf_b \Delta C(b)$.
First of all, this means that the projection of $u_2$ coordinate of
of the normal vector of the line that can be tangent to $\NR$ has to be
negative. If that projection is positive, that line can only intersect the
set $\NR$. Second, the maximum slope of the normal vector
to the tangent line is $\sup_b\Delta P(b)$ and the minimum is
$\inf_b\Delta P(b)$. Any line with a steeper slope will intersect the
set $\NR$.

Now we consider the restriction of the set $\NR$ via Assumption \ref{assume3}.
The additional restriction on what is rationalizable  does not change its convexity, but it makes
the resulting set bounded. Denote this bounded subset of $\NR$ by $\NR_B$. The
following theorem characterizes the structure of the support function
of this set.
\begin{theorem}
Let $\bar v$ correspond to the furthest point of the set
$\NR_B$ from the $\epsilon$ axis. Suppose that
$\bar b$ is such that $\bar v \Delta P(\bar b)-\bar\epsilon=\Delta C(\bar b)$.
Also suppose that $\underline{b}$ is the point of intersection of the boundary of the set
$\NR$ with the vertical axis $v=0$.
Under Assumptions \ref{assume1}, \ref{assume2}, and \ref{assume3} the support function
of $\NR_B$ is
$$
h(\NR_B,u)=\left\{
\begin{array}{ll}
h(\NR,u),&\,\mbox{if}\,u_2<0,\,\Delta P(\underline{b})<u_1/|u_2|<\Delta P(\bar b),\\
u_1\bar v +u_2 \bar \epsilon,&\,\mbox{if}\,u_2<0,\,u_1/|u_2|>\Delta P(\bar b),\\
u_2 \underline{\epsilon},&\,\mbox{if}\,u_2<0,\,\Delta P(\underline{b})>u_1/|u_2|,\\
u_1\bar v +u_2 \bar \epsilon,&\,\mbox{if}\,u_1,u_2\geq 0,\\
u_2 \bar \epsilon,&\,\mbox{if}\,u_1<0,\,u_2\geq 0,\\
\end{array}
\right.
$$
\end{theorem}

\begin{proof}
We begin with $u_2<0$. In this case when $\Delta P(\underline{b})<u_1/|u_2|<\Delta P(\bar b)$,
then the corresponding normal vector $u$ is on the part of the boundary
of the set $\NR_B$ that coincides with the boundary of $\NR$. And thus
for this set of normal vectors $h(\NR_B,u)=h(\NR,u)$.
\par
Suppose that $u_2<0$ while $u_1/|u_2| \geq \Delta P(\bar b)$. In this
case the support hyperplane is centered at point $(\bar v, \bar \epsilon)$
for the entire range of angles of the normal vectors. Provided that
the equation for each such a support hyperplane corresponds to
$\frac{u_1}{|u_2|}v-\epsilon=c$ and each needs pass through $(\bar v, \bar \epsilon)$,
then the support function is $h(\NR_B,u)=u_1 \bar v+u_2 \bar \epsilon$.
\par
Suppose that $u_2<0$ and $u_1/|u_2|<\Delta P(\underline{b})$. Then
the support hyperplanes will be centered at point $(0,\underline{v})$.
The corresponding support function can be expressed as
$h(\NR_B,u)=u_2 \underline{\epsilon}$.
\par
Now suppose that $u_2 \geq 0$ and $u_1>0$. Then the support hyperplanes
are centered at $(\bar v, \bar \epsilon)$ and again the support function
will be $h(\NR_B,u)=u_1 \bar v+u_2 \bar \epsilon$.
\par
Finally, suppose that $u_2 \geq 0$ and $u_1 \leq 0$. Then the
support hyperplanes are centered at $(0,\bar\epsilon)$. The corresponding
support function is $h(\NR_B,u)=u_2\bar \epsilon$.
\end{proof}

\section{Inference for Rationalizable Set}\label{sec:statistical}

Note that to construct the support function (and thus fully characterize
the set $\NR_B$) we only need to evaluate the function
$\Delta C\left(\Delta P^{-1}\left(\cdot\right)\right)$. It is one-dimensional
function and can be estimated from the data via direct simulation.
We have previously noted that our goal is to characterize the
distance between the true set $\NR_B$ and the set $\widehat{\NR}$
that is obtained from \etedit{subsampling} the data. Denote
$$
f(\cdot)=\Delta C\left(\Delta P^{-1}\left(\cdot\right)\right)
$$
and let $\widehat{f}(\cdot)$ be its estimate from the data.
The set $\NR_B$ is characterized by its support function
$h(\NR_B,u)$ which is determined by the exogenous upper bound
on the error $\bar \epsilon$, the intersection point
of the boundary of $\NR$ with the vertical axis $(0,\underline{\epsilon})$,
and the highest rationalizable value $\bar v$.
\par
We note that the set $\NR$ lies inside the shifted cone defined
by half-spaces $v\sup_b \Delta P(b)-\epsilon \leq \sup_b\Delta C(b)$
and $v\inf_b \Delta P(b)-\epsilon \leq \inf_b\Delta C(b)$. Thus the
value $\bar v$ can be upper bounded by the intersection of the line
$v\sup_b \Delta P(b)-\epsilon = \sup_b\Delta C(b)$ with $\epsilon=\bar\epsilon$.
The support function corresponding to this point can therefore be
upper bounded by $|u_2|\sup_b\Delta C(b)$.
\par
Then we notice that
$$
d_H(\widehat{\NR}_B,\,\NR_B)=\sup\limits_{\|u\|=1}|h(\widehat{\NR}_B,u)-h(\NR_B,u)|.
$$
For the evaluation of the sup-norm of the difference between the
support function, we split the circle $\|u\|=1$ to the areas where the
support function is linear and non-linear determined by the function
$f(\cdot)$. For the non-linear segment the sup-norm can be upper bounded by
the global sup-norm
\begin{align*}
\sup\limits_{u_2<0,\,\Delta P(\underline{b})<u_1/|u_2|<\Delta P(\bar b)}|h(\widehat{\NR}_B,u)-h(\NR_B,u)|
&\leq \sup\limits_{\|u\|=1}\left||u_2|
\widehat{f}\left(\frac{u_1}{|u_2|}\right)-|u_2|
{f}\left(\frac{u_1}{|u_2|}\right)\right|\\
& \leq \sup\limits_{z}\left|
\widehat{f}\left(z\right)-
{f}\left(z\right)\right|.
\end{align*}
For the linear part, provided that the value $\bar \epsilon$ is fixed,
we can evaluate the norm from above by
\begin{align*}
\sup\limits_{\|u\|=1}|u_1\sup_b\Delta \widehat{C}(b)-u_1 \Delta C(b)| &\leq |\sup_b\Delta \widehat{C}(b)- \Delta C(b)|
\leq \sup\limits_{z}\left|
\widehat{f}\left(z\right)-
{f}\left(z\right)\right|.
\end{align*}

Thus
we can evaluate
\begin{equation}\label{bound}
\textstyle{d_H(\widehat{\NR}_B,\,\NR_B) \leq \sup\limits_{z}\left|
\widehat{f}\left(z\right)-
{f}\left(z\right)\right|.}
\end{equation}
Thus, the bounds that can be established for estimation of function
$\widehat{f}$ directly imply the bounds for estimation of the Hausdorff
distance between the estimated and the true sets $\NR_B$.
We assume that a sample of size $N=n \times T$ is available (where $n$ is the
number of auctions sampled per period and $T$ is the number of periods). We
now establish the general rate result for the estimation of the set
$\NR_B$.

\begin{theorem}
Suppose that function $f$ has derivative up to order $k \geq 0$ and for some $L \geq 0$
$$
|f^{(k)}(z_1)-f^{(k)}(z_2)| \leq L|z_1-z_2|^{\alpha}.
$$
Under Assumptions \ref{assume1}, \ref{assume2}, and \ref{assume3} we have
$$
d_H(\widehat{\NR}_B,\,\NR_B) \leq O((N^{-1}\log\,N)^{\gamma/(2\gamma+1)}),\;\;\gamma=k+\alpha.
$$
\end{theorem}
\paragraph{Remark}
The theorem makes a further assumption the function $f$ is $k$ times differentiable, and satisfies a Lipschitz style bound with parameter $L \geq 0$ and exponent $\alpha$. We note that this theorem if we take he special case of $k=0$, the theorem does not require differentiability of functions $\Delta P(\cdot)$ and $\Delta C(\cdot)$. If these functions are Lipschitz, the condition of the theorem is satisfied with $k=0$ and $\alpha=1$, and the theorem provides a $O((N^{-1}\log\,N)^{1/3})$
convergence rate for the estimated set $\NR$.
\begin{proof}
By (\ref{bound}) the error in the estimation of the set $\NR_B$
is fully characterized by the uniform error in the estimation
of the single-dimensional function $f(\cdot)$. In part (i) of the
Theorem our assumption is that we estimate function $f(\cdot)$ from the
class of convex functions. In part (ii) of the Theorem our assumption
is that we estimate the function from the class of $k$ times differentiable
single-dimensional functions. We now use the results for {\it optimal global convergence rates}
for estimation of functions from these respective classes. We note that these rates
do not depend on the particular chosen estimation method and thus provide
a global bound on the convergence of the estimated set $\widehat{\NR}_B$
to the true set.
By \cite{stone1982}, the global uniform convergence rate for estimation of the
unknown function with $k$ derivatives where $k$-th derivative is H\"{o}lder-continuous
with constant $\alpha$ is $(N^{-1}\log\,N)^{\gamma/(2\gamma+1)})$ with $\gamma=k+\alpha$. That delivers
our statement.
\end{proof}

We note that this theorem does not require differentiability of functions
$\Delta P(\cdot)$ and $\Delta C(\cdot)$. For instance, if these functions
are Lipschitz, the theorem provides a $O((N^{-1}\log\,N)^{1/3})$
convergence rate for the estimated set $\NR$.

\section{Data Analysis}\label{sec:data}

We apply our approach to infer the valuations \dnedit{and regret constants} of a set of advertisers from the search Ads system of Bing.
\dnedit{Our focus is on advertisers who change their bids frequently (up to multiple bid changes per hour)
and thus are more prone to use algorithms for their bid adjustment instead of changing those
bids manually. Each advertiser
corresponds to an ``account". Accounts have multiple listings corresponding to individual ads.
The advertisers can set the bids for each listing within the account separately.}
We examine nine high frequency accounts from the same industry sector for a period of a week and apply our techniques to all the listings within each account. \etnedit{The considered market is highly dynamic where new advertising campaigns are launched on the daily basis (while there is also a substantial amount of experimentation on the auction platform side, that has a smaller contribution to the uncertainty regarding the market over the uncertainty of competitors' actions.)}
We focus on analyzing
\dnedit{the bid dynamics across the }
listings within the same account as they are most probably instances of the same learning algorithm. Hence the only thing that should be varying for these listings is the bidding environment and the valuations. Therefore, statistics across such listings, capture in some sense the statistical behavior of the learning algorithm when the environment and the valuation per bid is drawn from some distribution.

\paragraph{Computation of the Empirical Rationalizable Set}
We first start by briefly describing 
\dnedit{the procedure that we constructed}
to compute the set $\NR$ for a single listing. We assumed that bids and values have a hard absolute upper bound and since they are constrained to be multiples of pennies, the strategy space of each advertiser is now \vsedit{a finite set,} 
rather than the set $\R_+$. Thus for each possible deviating bid $b'$ in this \vsedit{finite set}
we compute the $\Delta P(b')$ and $\Delta C(b')$ 
\dnedit{for each listing.}
We then discretize the space of additive errors. For each additive error $\epsilon$ in this discrete space we use the characterization in Lemma \ref{lem:upper-lower} to compute the upper and lower bound on the value for this additive error. This involves taking a maximum and a minimum over a finite set of alternative bids $b'$. We then look at the smallest epsilon $\epsilon_0$, where the lower bound is smaller than the upper bound and this corresponds to the smallest rationalizable epsilon. For every epsilon larger than $\epsilon_0$ we plot the upper bound function and the lower bound function.

An example of the resulting set $\NR$ for a high frequency listing of one of the advertisers we analyzed is depicted in Figure \ref{fig:nr-set}. From the figure, \dnedit{we} observe that the right listing \dnedit{has a }
higher regret than the left one. Specifically, the smallest rationalizable additive error is further from zero. Upon \dnedit{the } examination of the bid \dnedit{change} pattern, the right listing in the Figure was in a more transient period where the bid of the advertiser was increasing, hence this could have been a period were the advertiser was experimenting. The bid pattern of the first listing was more stable.

\begin{figure}
\centering
\subfigure{
\includegraphics[scale=.30]{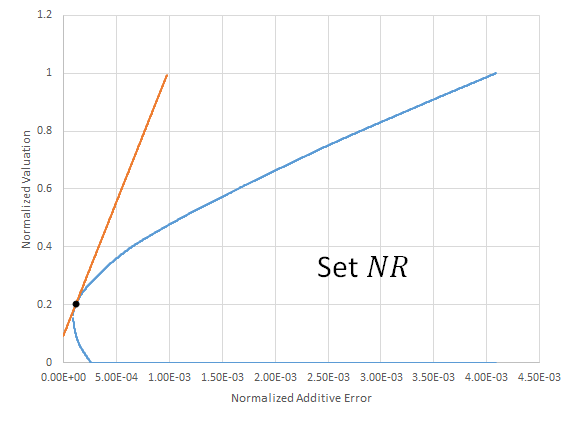}
}
\subfigure{
\includegraphics[scale=.30]{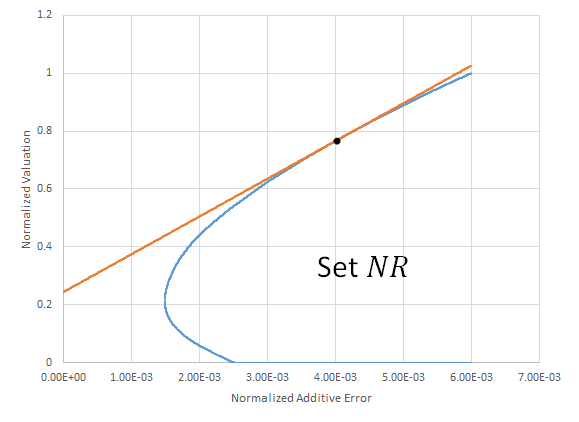}
}
\caption{$\NR$ set for two listings of a high-frequency bid changing advertiser. Values are normalized to $1$. The tangent line selects our  point prediction. }\label{fig:nr-set}
\end{figure}

\paragraph{Point Prediction: Smallest Multiplicative Error}
\label{subsec:multi}
Since the two dimensional rationalizable set $\NR$ is hard to depict as summary statistic from a large number of listings, we instead opt for a point prediction and specifically we compute the point of the $\NR$ set that has the smallest \dnedit{regret }
error.

Since the smallest possible additive error is hard to compare across listings, we instead pick the smallest multiplicative error that is rationalizable, i.e. a pair $(\delta,v)$ such that:
\begin{equation}\label{eqn:eps-mult-regret}
\textstyle{\forall b'\in B_i: \frac{1}{T} \sum_{t=1}^{T} U_i\left(\vec{b}^t;\theta^t, v_i\right) \geq (1-\delta)\frac{1}{T} \sum_{t=1}^{T} U_i\left(b',\vec{b}_{-i}^t;\theta^t, v_i\right)}
\end{equation}
and such that $\delta$ is the smallest possible \dnedit{value} that admits such a value \dnedit{per click} $v$. 
\dnedit{D}enote 
$P^0 = \frac{1}{T}\sum_{t=1}^{T}P_i^t(\vec{b}^t)$ the observed average probability of click of the advertiser and with $C^0=\frac{1}{T}\sum_{t=1}^{T}C_i^t(\vec{b}^t)$ the observed average cost, then by simple re-arrangements the latter constraint is equivalent to:
\begin{equation}\label{eqn:reg-const}
\textstyle{\forall b'\in B_i: v\Delta P(b) \leq \Delta C(b)+ \frac{\delta}{1-\delta}\left( v P^0 - C^0\right)}
\end{equation}
By comparing \dnedit{this result }to Equation \eqref{eqn:halfplanes}, a multiplicative error of $\delta$ corresponds to an additive error of $\epsilon=\frac{\delta}{1-\delta}\left( v P_i^0 - C_i^0\right)$.

Hence, one way to compute the feasible region of values for a multiplicative error 
$\delta$ from the $\NR$ set that we estimated is to draw a line of $\epsilon=\frac{\delta}{1-\delta}\left( v P_i^0 - C_i^0\right)$. The two points where this line intersects the $\NR$ set correspond to the upper and lower bound on the valuation. Then the smallest multiplicative error $\delta$, corresponds to the line that simply touches the $\NR$ set, i.e. the smallest $\delta$ for which the intersection of the line $\epsilon=\frac{\delta}{1-\delta}\left( v P_i^0 - C_i^0\right)$ is non-empty. This line is depicted in orange in Figure \ref{fig:nr-set} and the point appears with a black mark.

Computationally, we estimate this point by binary search on the values of $\delta\in [0,1]$, until the upper and lower point of $\delta$ in the binary search is smaller than some pre-defined precision. Specifically, for each given $\delta$, the upper and lower bound of the values implied by the constraints in Equation \ref{eqn:reg-const} is:
\begin{equation*}
\max_{b': (1-\delta)\Delta P(b') - \delta P^0>0} \frac{(1-\delta) \Delta C(b') - \delta C^0}{ (1-\delta)\Delta P(b') - \delta P^0} \leq v\leq \min_{b': (1-\delta)\Delta P(b') - \delta P^0>0} \frac{(1-\delta) \Delta C(b') - \delta C^0}{ (1-\delta)\Delta P(b') - \delta P^0}
\end{equation*}
If these inequalities have a non-empty intersection for some value of $\delta$, then they have a non-empty intersection for any larger value of $\delta$ (as is implied by our graphical interpretation in the previous paragraph).

Thus we can do a binary search over the values of $\delta\in (0,1)$. At each iteration we have an upper bound $H$ and lower bound $L$ on $\delta$. We test whether $\delta= (H+L)/2$ gives a non-empty intersection in the above inequalities. If it does, then we decrease the upper bound $H$ to $(H+L)/2$, if it doesn't then we increase the lower bound $L$ to $(H+L)/2$. We stop whenever $H-L$ is smaller than some precision, or when the implied upper and lower bounds on $v$ from the feasible region for $\delta=H$, are smaller than some precision.

The value that corresponds to the smallest rationalizable multiplicative error $\delta$ can be viewed as a point prediction for the value of the player. It is exactly the value that corresponds to the black dot in Figure \ref{fig:nr-set}. Since this is a point of the $\NR$ set the estimation error of this point from data has at least the same estimation error convergence properties as the whole $\NR$ set that we derived in Section \ref{sec:statistical}.

\paragraph{Experimental Findings} We compute the pair of the smallest rationalizable multiplicative error $\delta^*$ and the corresponding predicted value $v^*$ for every listing of each of the nine accounts that we analyzed. In Figure \ref{fig:histogram1}, on the right we plot the distribution of the smallest non-negative rationalizable error over the listings of a single account. \etnedit{Different listings of a single account are driven by the same algorithm, hence we view this plot as the ``statistical footprint'' of this algorithm}.  We observe that all accounts have a similar pattern in terms of the smallest rationalizable error: almost 30\% of the listings within the account can be rationalized with an almost zero error, i.e. $\delta^*<0.0001$. \vsnedit{We note that regret can also be negative, and in the figure we group together all listings with a negative smallest possible error. This contains 30\% of the listings.} The empirical distribution of the regret constant $\delta^*$ for the \etnedit{remaining 70\% } of the listings 
\dnedit{is close to the uniform distribution on }
$[0,.4]$. Such a pattern suggests that half of the listings of an advertiser have reached a stable state, while a large fraction of them are still in a learning state.

\begin{figure}
\centering
\subfigure{
\includegraphics[scale=.30]{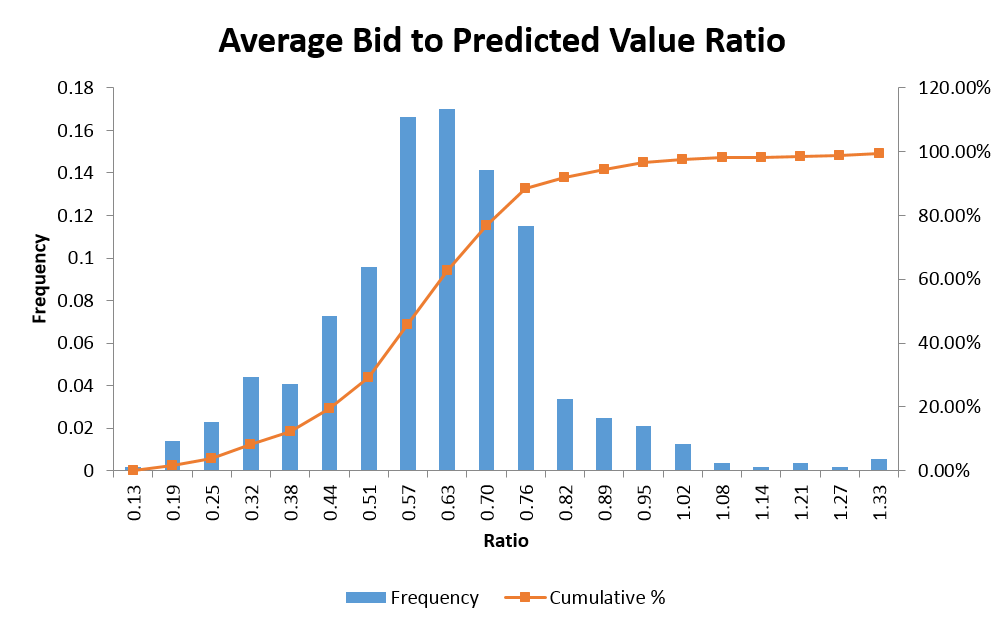}
}
\subfigure{
\includegraphics[scale=.30]{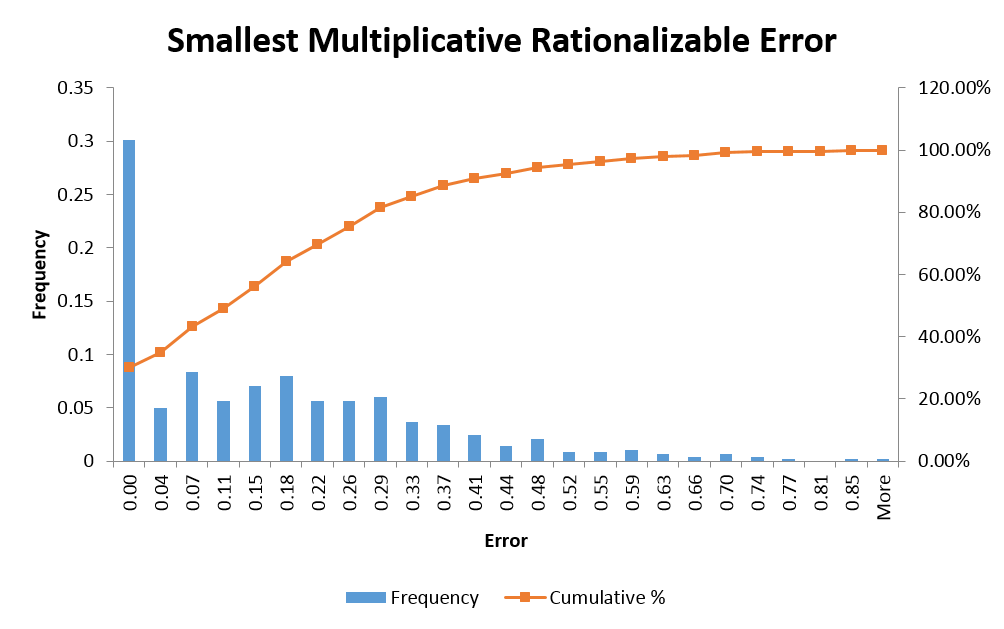}
}
\caption{Histogram of the ratio of predicted value over the average bid in the sequence and the histogram of the smallest \vsnedit{non-negative rationalizable multiplicative error $\delta$ (the smallest bucket contains all listings with a non-positive smallest rationalizable error)}.}\label{fig:histogram1}
\end{figure}

We also analyze \dnedit{the amount of relative bid shading. }
\dnedit{It has been previously observed that bidding in the GSP auctions is not truthful.
Thus we can empirically measure the difference between the bids and estimated values
of advertisers associated with different listings. }
Since the bid of a listing is not constant in the period of the week that we analyzed, we use the average bid as proxy for the bid of the advertiser. Then we compute the ratio of the average bid over the predicted value for each listing. We plot the distribution of this ratio over the listings of a \etnedit{typical} account in the left plot of Figure \ref{fig:histogram1}. We observe that based on our prediction, advertisers consistently bid on average around 60\% of their value. Though the amount of bid shading does have some variance,
\dnedit{the empirical distribution of the ratio of average bid and the estimated value is close to normal
distribution }
with mean around 60\%.

\begin{wrapfigure}{l}{0.35\textwidth}
\begin{center}
\includegraphics[scale=.25]{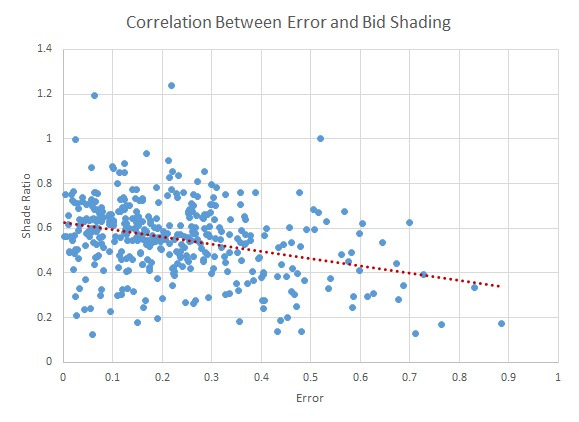}
\end{center}
\caption{Scatter plot of pairs $(v^*,\delta^*)$ for all listings of a single advertiser.}\label{fig:scatter}
\end{wrapfigure}

Interestingly, we observe that based on our prediction, advertisers consistently bid on average around 60\% of their value. Though the amount of bid shading does have some variance,
\dnedit{the empirical distribution of the ratio of average bid and the estimated value is close to normal
distribution }
with mean around 60\%.\footnote{ We do observe that for a very small number of listings within each account our prediction implies a small amount of overbidding on average. These outliers could potentially be because our main assumption that the value remains fixed throughout the period of the week doesn't hold for these listings due to some market shock.}

Figure \ref{fig:histogram1} \etnedit{depicts} the ratio distribution 
for one 
account\etndelete{ that we analyzed}. We 
give similar plots for all other accounts that we analyzed in the \etnedit{full version of the paper.}

Last we also analyze whether there is any correlation between the smallest rationalizable error and the amount of underbidding for listings that seem to be in a learning phase (i.e. have $\delta^*$ greater than some tiny amount). We present a scatter plot of the pairs of $(\delta^*,v^*)$ in Figure \ref{fig:scatter} for a single account and for listings that have $\delta^*>0.0001$\etndelete{ (see Figure \ref{fig:scatter2} in the Appendix for the remaining accounts)}. Though there does not seem to be a significant correlation we consistently observe a 
\dnedit{small} correlation: listings with higher error shade their bid more.
%
%

%

\bibliographystyle{acmsmall}
\bibliography{bib-AGT}


%

\end{document}